\documentclass[11pt]{article}

\usepackage{amsmath,amssymb}
\usepackage{QED}
\usepackage[all]{xy}
\usepackage{pgf}
\usepackage{tikz}

\newtheorem{definition}{Definition}
\newtheorem{proposition}[definition]{Proposition}
\newtheorem{lemma}[definition]{Lemma}
\newtheorem{theorem}[definition]{Theorem}
\newtheorem{corollary}[definition]{Corollary}

\newtheorem{observation}[definition]{Observation}

\usepackage{anysize}
\marginsize{2.5cm}{2.5cm}{2.5cm}{2.5cm}

\author{St\'ephane Le Roux\\Universit\'e Libre de Bruxelles}

\title{Infinite subgame perfect equilibrium\\in the Hausdorff difference hierarchy}

\begin{document}
\maketitle

\begin{abstract}
\noindent Subgame perfect equilibria are specific Nash equilibria in perfect information games in extensive form. They are important because they relate to the rationality of the players. They always exist in infinite games with continuous real-valued payoffs, but may fail to exist even in simple games with slightly discontinuous payoffs. This article considers only games whose outcome functions are measurable in the Hausdorff difference hierarchy of the open sets (\textit{i.e.} $\Delta^0_2$ when in the Baire space), and it characterizes the families of linear preferences such that every game using these preferences has a subgame perfect equilibrium: the preferences without infinite ascending chains (of course), and such that for all players $a$ and $b$ and outcomes $x,y,z$ we have $\neg(z <_a y <_a x \,\wedge\, x <_b z <_b y)$. Moreover at each node of the game, the equilibrium constructed for the proof is Pareto-optimal among all the outcomes occurring in the subgame. Additional results for non-linear preferences are presented.
\end{abstract}


\section{Introduction}

\noindent Game theory is the theory of competitive interactions between agents having different interests. Until the late 1960's an agent would usually represent a human or group of humans, when game theory was mainly meant for economics and political science. Then game theory was also applied to evolutionary biology~\cite{MP73} and to theoretical computer science~\cite{BL69}, especially to system verification and system synthesis (against given specifications). Classically, the verification or synthesis problem is represented as a game with two players: the system trying to win the game by meeting the specifications, and the environment trying to win the game by preventing the system from doing so. The two players play this game in turn and deterministically on a finite or infinite directed graph, and the key notion is that of winning strategy. For a decade, though, computer scientists such as Ummels~\cite{Ummels06} have been considering multi-player games to represent more complex verification or synthesis problems, \textit{e.g.} relating to distributed systems. The notion of winning strategy is specific to two-player win-lose games, but in a multi-player setting it may be replaced with a faithful extension, namely the famous notion of (pure) Nash equilibrium. It does not only accommodate more than two players, it also allows for refined/quantitative objectives.

The deterministic turn-based games on graphs may be unfolded, usually without much loss of information, into deterministic turn-based games on finite or infinite trees, which have been widely studied in game theory. It is one reason why this article focuses on perfect information games in extensive form (\textit{i.e.} played on trees) and their deterministic strategies, unless otherwise stated.  

Kuhn~\cite{Kuhn53} proved the existence of Nash equilibrium (NE) in finite games with real-valued payoffs. His proofs uses backward induction and constructs a special kind of NE that was later called subgame perfect equilibrium (SPE) by Selten~\cite{Selten65}. An extension of Kuhn's result~\cite{SLR09} characterizes the preferences that always yield finite games with NE/SPE: the acyclic preferences. Also, Escard\'o and Oliva~\cite{EO10} studied generalizations of backward induction in possibly infinite yet well-founded game trees, \textit{i.e.} without infinite plays. The SPE have nice extra properties and are usually preferred over the more general NE: for psychology in a broad sense an SPE amounts to the absence of empty threats, and for system engineering in a broad sense it amounts to stability of a system regardless of the initial state.

The concept of infinite horizon is convenient in economics and also in computer science, \textit{e.g.}, for liveness. Gale and Stewart~\cite{GS53} studied infinite two-player win-lose games, but backward induction is no longer applicable since there may not be leaves to start the induction at. Nevertheless, they proved that if the winning set of each player is open or closed (with the usual topology), one player has a wining strategy. This result was extended by Wolfe~\cite{Wolfe55} for $\Sigma^0_2$ and $\Pi^0_2$ sets, then by other people to more complex sets, and eventually by Martin for Borel~\cite{Martin75} and even quasi-Borel~\cite{Martin90} sets. This is called (quasi-)Borel determinacy.

Mertens and Neymann~\cite[p. 1567]{Mertens87} found that Borel determinacy can be used to show existence of $\epsilon$-NE in infinite games with bounded Borel-measurable real-valued payoffs. By generalizing their technique, an abstract result about point-classes and determinacy~\cite{SLR13} implies a characterization of the preferences that always yield games with NE, in games with (quasi-)Borel-measurable outcome functions, countably many outcomes, and an arbitrary cardinality of players: the preferences whose inverses are well-founded. Then it was shown~\cite{LP14} that two-player antagonistic games (\textit{i.e.} abstract zero-sum games) with finitely many outcomes and (quasi-)Borel-measurable outcome function have SPE.

When the outcome function is a continuous real-valued payoff function, Fudenberg and Levine~\cite{FL83} showed that there is always an SPE in multi-player games. Similar results were obtained recently via an abstract and uniform density argument~\cite{LP14}. The continuity assumption may be slightly relaxed if one is willing to accept approximate SPE. Indeed existence of $\epsilon$-SPE was proved for lower-semicontinuous~\cite{FKMSSV10} and upper-semicontinuous~\cite{PS11} payoffs.

However, when the real-valued payoff function is discontinuous enough and the preferences are not antagonistic, there may be no ($\epsilon$-)SPE, as in the following example which is similar to~\cite[Example 3]{SV03}. Let $a$ and $b$ be two players with preferences $z <_a y <_a x$ and $x <_b z <_b y$. They are alternatively given the possibility to stop and yield outcomes $y$ and $z$, respectively, but the outcome is $x$ if no one ever stops.

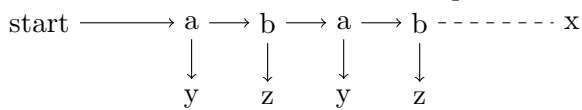
\begin{tikzpicture}[node distance=1cm]
  \node(s){start};
  \node(s1)[right of = s]{};
  \node(a1)[right of = s1]{a};
  \node(b1)[right of = a1]{b};
  \node(a2)[right of = b1]{a};
  \node(b2)[right of = a2]{b};
  \node(inf1)[right of = b2]{};
  \node(inf)[right of = inf1]{x};

   \node(y1)[below of = a1]{y};
   \node(z1)[below of = b1]{z};
   \node(y2)[below of = a2]{y};
   \node(z2)[below of = b2]{z};
    
  \draw [->] (s) -- (a1);
  \draw [->] (a1) -- (b1);
  \draw [->] (b1) -- (a2);
  \draw [->] (a2) -- (b2); 
  \draw [->] (a1) -- (y1);
  \draw [->] (b1) -- (z1);
  \draw [->] (a2) -- (y2); 
  \draw [->] (b2) -- (z2); 
  \draw [dashed] (b2) -- (inf);
  
 \end{tikzpicture}

\noindent In addition, a real-valued two-player game~\cite{FKMSSSV14} was recently designed with the following features: it has a similar preference pattern as in the example above, and it has no $\epsilon$-SPE for small enough $\epsilon$ even when the players are allowed to use mixed strategies at every node of the game tree. All this shows that Mertens~\cite{Mertens87} was right when writing that "Subgame perfectness is a completely different issue" from NE. This article solves the issue partially, and the contribution is two-fold. First, it characterizes the linear preferences that always yield SPE in games with outcome functions in the Hausdorff difference hierarchy of the open sets: the preferences that are void of infinite ascending chains (of course) and of the \emph{SPE killer} from the example above. Said otherwise, if $\neg(z <_a y <_a x \,\wedge\, x <_b z <_b y)$ holds for all players $a$ and $b$ and all outcomes $x$, $y$ and $z$ of a multi-player game with outcome function in the difference hierarchy (and preferences without infinite ascending chains), the game has an SPE, and even a \emph{global-Pareto} one, as in Definition~\ref{def:sp-ip-gpe}. Second contribution, the characterization still holds for two-player games with strict weak order preferences, but no longer for three-player games. (Strict weak orders are important since they are an abstraction of the usual preference order over the real-valued payoff functions.)

Section~\ref{sec:tb} consists of definitions; Section~\ref{sec:mp-lop} proves the characterization for many players and linear preferences; Section~\ref{sec:tp-swop} proves the characterization for two players and strict weak order preferences.

\paragraph{Related works} Characterizing the preferences that guarantee existence of NE in interesting classes of games is not a new idea: earlier than the two examples above (\cite{SLR09}, \cite{SLR13}), Gimbert and Zielonka~\cite{GZ05} "characterise the family of payoff mappings for which there always exist optimal positional strategies for both players" in some win-lose games played on finite graphs. Also, \cite{SLR14} characterizes  the preferences that guarantee existence of NE in determined, countable two-player game forms. 

The notion of SPE has also been studied in connection with system verification and synthesis, at low levels of the Borel hierarchy: in~\cite{Ummels06} with qualitative objectives, in~\cite{BBDG12} with quantitative objective for reachability, and in~\cite{BBMR15} with quantitative objectives and a weak variant of SPE.

Finally, some specific infinite games in extensive form (such as the dollar auction) and especially their SPE have been studied using co-algebraic methods in ~\cite{LP12} and \cite{AW12}.

\section{Technical background}\label{sec:tb}

\noindent The games in this article are built on infinite trees, which may be defined as prefix-closed sets of finite sequences. The elements of a tree are called nodes. Intuitively, a node represents both a position in the tree and the only path from the root to this position.

\begin{definition}[Tree]
Let $\Sigma$ be a set. 
\begin{itemize}
\item $\Sigma^*$ ($\Sigma^\omega$) is the set of finite (infinite) sequences over $\Sigma$, and a tree over $\Sigma$ is a subset $T$ of $\Sigma^*$ such that $\gamma \sqsubseteq \delta$ (\textit{i.e.} $\gamma$ is a prefix of $\delta$) and $\delta\in T$ implies $\gamma\in T$.
\item For a node $\gamma$ in a tree $T$, let $\mathrm{succ}(T,\gamma) := \{\delta\in T\,\mid\, \gamma\sqsubseteq\delta\,\wedge\, |\delta| = |\gamma|+1\}$, where $|\gamma|$ is the length of $\gamma$.
\item A tree $T$ is pruned if $\mathrm{succ}(T,\gamma) \neq \emptyset$ for all $\gamma\in T$. 
\item Let $T$ be a tree over $\Sigma$. The set $[T]$ is made of the infinite paths of $T$, namely the elements of $\Sigma^\omega$ whose every finite prefix is in $T$.
\item Let $T$ be a tree over $\Sigma$. For $\gamma\in T$ let $T_\gamma := \{\delta \in \Sigma^*\,\mid\, \gamma \delta \in T\}$.
\end{itemize}
\end{definition}

\noindent In this article the outcomes of a game correspond to some partition of the infinite paths of the game tree, and the subsets of the partition are restricted to the Hausdorff difference hierarchy of the open sets. This hierarchy is defined in, \textit{e.g.}, \cite[22.E]{Kechris95}, but a probably-folklore result~\cite[Section 2.4]{LP15} gives an equivalent presentation, which is in turn rephrased in Definition~\ref{def:qdh} below. These new definitions facilitate the proofs by induction in this article. Then, the Hausdorff-Kuratowski theorem (see, \textit{e.g.}, \cite[Theorem 22.27]{Kechris95}) implies that, in the Baire space (\textit{i.e.} $\mathbb{N}^\mathbb{N}$), the difference hierarchy is equal to $\Delta^0_2$, the sets that are both countable unions of closed sets and countable intersections of open sets. This equality tells how low the difference hierarchy lies in the Borel hierarchy.

For every pruned tree $T$, let $\{\gamma [T_\gamma]\,\mid\,\gamma\in T\}$ be a basis for the open subsets of $[T]$. Definition~\ref{def:cdh} below is a special case of a more general definition that can be found, \textit{e.g.}, in \cite[22.E]{Kechris95}.

\begin{definition}[Difference hierarchy]\label{def:cdh}
Let $T$ be a pruned tree, $\theta > 0$ be a countable ordinal, and $(A_\eta)_{\eta < \theta}$ be an increasing sequence of open subsets of $[T]$. $D_\theta((A_\eta)_{\eta < \theta})$ is defined as below. 
\begin{align*}
x \in D_\theta((A_\eta)_{\eta < \theta}) \,:\Leftrightarrow\, & x \in \cup_{\eta < \theta}A_\eta \textrm{ and the least } \eta < \theta \textrm{ with } x\in A_\eta \textrm{ has parity opposite to that of } \theta.
\end{align*}
And $D_\theta([T]) := \{D_\theta((A_\eta)_{\eta < \theta})\mid \forall \eta < \theta, A_\eta \textrm{ is an open subset of }[T]\}$.
\end{definition}

\begin{observation}\label{obs:cdh}
$D_{\theta +1}((A_\eta)_{\eta < \theta +1}) = A_\theta \setminus D_\theta((A_\eta)_{\eta < \theta})$
\end{observation}

\noindent Definition~\ref{def:qdh}, Lemma~\ref{lem:cdh-incl-adh}, and Proposition~\ref{prop:cdh-eq-adh} relate to \cite[Section 2.4]{LP15}.

\begin{definition}[Quasi-difference sets]\label{def:qdh}
Let $T$ be a pruned tree. The set $\mathcal{D}([T])$ is defined by transfinite induction below.
\begin{itemize}
\item Every open set of $[T]$ is in $\mathcal{D}([T])$.
\item Let $(\gamma_i)_{i\in I}$ be pairwise non-comparable nodes in $T$, and for all $i\in I$ let $D_i\in \mathcal{D}([T])$ be such that $D_i \subseteq \gamma_i[T_{\gamma_i}]$. Then $\cup_{i\in I} D_i \in \mathcal{D}([T])$.
\item The complement in $[T]$ of a set in $\mathcal{D}([T])$ is also in $\mathcal{D}([T])$.
\end{itemize}
\end{definition}

\noindent One can prove Observation~\ref{obs:qdh-inter-ball} below by induction along Definition~\ref{def:qdh}.

\begin{observation}\label{obs:qdh-inter-ball}
Given a node $\gamma$ of a pruned tree $T$, and $A \in \mathcal{D}([T])$, $\gamma [T_\gamma] \cap A \in \mathcal{D}([T])$.
\end{observation}

\begin{lemma}\label{lem:cdh-incl-adh}
$D_\theta \subseteq \mathcal{D}([T])$ for all non-zero countable ordinal $\theta$ and all pruned tree $T$.
\begin{proof}
By induction on $\theta$, which holds for $\theta = 1$ since $D_1([T])$ is made of the open subsets of $[T]$. Let $\theta > 1$ be an ordinal and let $A \in D_\theta$, so $A = D_{\theta}((A_\eta)_{\eta < \theta})$ for some family $(A_{\eta})_{\eta < \theta}$ of open subsets of $[T]$ that is increasing for the inclusion. Every $A_\eta$ is open so it can be written $\cup_{i\in I}\gamma_{\eta,i}[T_{\gamma _{\eta,i}}]$ where $\gamma_{\eta,i}$ and $\gamma_{\eta,j}$ are not proper prefixes of one another (but are possibly equal) for all $\eta < \theta$ and $i,j\in I$. We can also require some minimality for all $\gamma_{\eta,i}$, more specifically $\gamma \sqsubset \gamma_{\eta,i} \Rightarrow \neg(\gamma[T_\gamma] \subseteq A_\eta)$. Let $F$ consists of the minimal prefixes among $\{\gamma_{\eta,i}\,\mid\,i\in I \wedge \eta < \theta\}$, and for all $\gamma\in F$ let $f(\gamma)$ be the least $\eta$ such that $\gamma = \gamma_{\eta,i}$ for some $i\in I$. So $f(\gamma) < \theta$ for all $\gamma \in F$, and $A \subseteq \cup_{\eta < \theta} A_\eta = \cup_{\gamma\in F}\gamma[T_\gamma]$, so $A = \cup_{\gamma\in F}A \cap \gamma[T_\gamma]$. Let $\gamma \in F$ and let us make a case disjunction to show that $A \cap \gamma[T_\gamma] \in \mathcal{D}([T])$. First case, $f(\gamma)$ and $\theta$ have the same parity, so $A \cap \gamma[T_\gamma] = D_{f(\gamma)}((A_\eta \cap \gamma[T_\gamma])_{\eta < f(\gamma)}) \in D_{f(\gamma)}([T])$, so $A \cap \gamma[T_\gamma] \in \mathcal{D}([T])$ by induction hypothesis. Second case, $f(\gamma)$ and $\theta$ have opposite parity, so 
\begin{align*}
A \cap \gamma[T_\gamma] \quad=\quad & D_{f(\gamma)+1}((A_\eta \cap \gamma[T_\gamma])_{\eta < f(\gamma)+1})\textrm{ by Definition~\ref{def:cdh}},\\
	=\quad & (A_{f(\gamma)}\cap \gamma[T_\gamma]) \setminus D_{f(\gamma)}((A_\eta \cap \gamma[T_\gamma])_{\eta < f(\gamma)})) \textrm{ by Observation~\ref{obs:cdh}},\\
	=\quad & \gamma[T_\gamma] \setminus D_{f(\gamma)}((A_\eta \cap \gamma[T_\gamma])_{\eta < f(\gamma)})) \textrm{ by definition of }F\ni\gamma\textrm{ and }f(\gamma),\\
	\in\quad & \mathcal{D}([T]) \textrm{ by induction hypothesis and Observation~\ref{obs:qdh-inter-ball}.} 
\end{align*}
\noindent Therefore $A\in \mathcal{D}([T])$.
\end{proof}
\end{lemma}

\noindent Proposition~\ref{prop:cdh-eq-adh} below shows that the quasi-difference sets coincide with the sets in the difference hierarchy for countable trees, just like quasi-Borel sets~\cite{Martin90} and Borel sets coincide on Polish spaces.

\begin{proposition}\label{prop:cdh-eq-adh}
$\mathcal{D}([T]) = \cup_{\theta < \omega_1}D_\theta([T])$ for all countable pruned tree $T$.
\begin{proof}
$\cup_{\theta < \omega_1}D_\theta \subseteq \mathcal{D}([T])$ was already proved in Lemma~\ref{lem:cdh-incl-adh}, so let $A\in \mathcal{D}([T])$. Let us prove that $A\in D_\theta([T])$ for some $\theta$, by induction on the definition of $\mathcal{D}([T])$. Base case, if $A$ is open, $A\in D_1([T])$.

Second case, let $(\gamma_i)_{i\in I}$ be pairwise non-comparable nodes in $T$, and let $A_i\in \mathcal{D}([T]) \cap \gamma_i[T_{\gamma_i}]$ for all $i\in I$, such that $A = \cup_{i\in I} A_i$. By induction hypothesis $A_i \in D_{\theta_i}$ for some $\theta_i < \omega_1$. Let $\theta := \sup_{i\in I} \theta_i$, so $\theta < \omega_1$ by countability of $T$, and $A_i \in D_{\theta}$ for all $i\in I$. For all $i\in I$ let $(A_{i,\eta})_{\eta < \theta}$ be an increasing sequence of open sets such that $A_i = D_{\theta}((A_{i,\eta})_{\eta < \theta})$. So $A = \cup_{i\in I} D_{\theta}((A_{i,\eta})_{\eta < \theta}) = D_{\theta}((\cup_{i\in I}A_{i,\eta})_{\eta < \theta})$, which shows that $A \in D_\theta([T])$.

Third case, $A$ is the complement of $B \in D_\theta([T])$ for some $\theta < \omega_1$. So $B = D_{\theta}((B_{\eta})_{\eta < \theta})$ for some increasing sequence of open sets of $[T]$. Let $B_{\theta} = [T]$, so $A = D_{\theta + 1}((B_{\eta})_{\eta < \theta + 1})$.
\end{proof}
\end{proposition}

\noindent Informally, a play of a game starts at the root of an infinite game tree and at each stage of the game the unique owner of the current node chooses a child of the node. The articles \cite{SLR13} and \cite{LP14} used game trees of the form $C^*$ because it was more convenient and done without much loss of generality, but this article works with general pruned trees because they will be cut in a non-uniform way. Moreover pruned trees are general enough, since leaves in infinite games can be simulated by the pseudo-leaves defined below.

\begin{definition}[Game, subgame, pseudo-leaf]\label{def:g-sg-pl}
An infinite game $g$ is a tuple $\langle A, T, d, O, v , (\prec_a)_{a\in A}\rangle$ complying with the following.
\begin{itemize}
\item $A$ is a non-empty set (of players).
\item $T$ is a non-empty pruned tree (of possible finite plays).
\item $d:T\to A$ (assigns a decision maker to each stage of the game).
\item $O$ is a non-empty set (of possible outcomes of the game). 
\item $v:[T]\to O$ (uses outcomes to value the infinite plays in the tree).
\item Each $\prec_a$ is a binary relation over $O$ (modelling the preference of player $a$).
\end{itemize}
For $\gamma\in T$, the subgame $g_\gamma$ is defined by $d\gamma : T_\gamma \to A$ such that $d_\gamma(\delta) := d(\gamma\delta)$ and by $v_\gamma : [T_\gamma] \to O$ such that $v_\gamma(p) := v(\gamma p)$. A $x$-pseudo-leaf of $g$ is a shortest node $\gamma\in T$ such that only the outcome $x$ occurs in $g_\gamma$.
\end{definition}

\begin{definition}[strategy profile, induced play, global-Pareto equilibrium]\label{def:sp-ip-gpe}
Let $g = \langle A, T, d, O, v , (\prec_a)_{a\in A}\rangle$ be a game. 

\begin{itemize}
\item A strategy profile is a function $s:T\to T$ such that $s(\gamma) \in \mathrm{succ}(T,\gamma)$ for all $\gamma\in T$. Let $S_g$ be the set of the strategy profiles for $g$. For $\gamma\in T$ and $s\in S_g$ the subprofile $s_\gamma: T_\gamma \to T_\gamma$ is defined by the equality $\gamma s_\gamma(\delta) = s(\gamma \delta)$.

\item For $\gamma\in T$ and $s\in S_g$, the play $p = p^\gamma(s)$ induced by $s$ at $\gamma$ is defined inductively by $p_0\dots p_{|\gamma|-1} := \gamma$ and $p_n := s(p_0\dots p_{n-1})$ for all $n > |\gamma|$.

\item A Nash equilibrium is a profile $s\in S_g$ such that
\[NE_g(s) := \forall s'\in S_g,\forall a\in A, \,\neg(v \circ p^\epsilon (s) \prec_a v \circ p^\epsilon(s') \,\wedge\, (\forall \gamma\in T, s(\gamma) \neq s'(\gamma) \Rightarrow d(\gamma) = a))\]
A subgame perfect equilibrium is a profile $s\in S_g$ such that $NE_{g_\gamma}(s_\gamma)$ for all $\gamma\in T$.

\item Let $O'\subseteq O$. One says that $x\in O'$ is Pareto-optimal in $O'$ if for all $y\in O'$ and $a\in A$ such that $x \prec_a y$ there exists $b\in B$ such that $y \prec_b x$. A global-Pareto Nash equilibrium (GP-NE) is an NE whose induced outcome is Pareto-optimal in the outcomes occurring in the underlying game. A GP-SPE is a profile that induces a GP-NE in every subgame.
\end{itemize}
\end{definition}

\noindent The proofs in this article do not build SPE by mere backward induction, but more generally by recursively refining rational behavioral promises. At each stage the refinement is optimal given the existing promises and regardless of the future ones. Since a promise not to choose a specific successor of a given node cannot be represent by a strategy profile, the more general notion of quasi-profile is defined below.

\begin{definition}[quasi profile]
Let $g = \langle A, T, d, O, v , (\prec_a)_{a\in A}\rangle$ be a game.
\begin{itemize}
\item A quasi profile is a multivalued function $q:T\multimap T$ such that $\emptyset \neq q(\gamma)\subseteq  \mathrm{succ}(\gamma)$ for all $\gamma\in T$. Let $Q_g$ be the set of the quasi profiles for $g$.  For $\gamma\in T$ and $q\in Q_g$ the sub-quasi-profile $q_\gamma: T_\gamma \multimap T_\gamma$ is defined by the equality $\gamma q_\gamma(\delta) = q(\gamma \delta)$.

\item For $\gamma\in T$ and $q\in Q_g$, the tree induced by $q$ starting at $\gamma$ is defined inductively by $\epsilon\in T_\gamma(q)$, where $\epsilon$ is the empty sequence, and $\delta\in T_\gamma(q) \Rightarrow  q_\gamma(\delta) \subseteq T_\gamma(q)$.

\item Let $q\in Q_g$. Let $(\gamma_i)_{i\in I}$ be the nodes of $T$ such that $\gamma_i \notin q(\gamma)$ for all $\gamma\in T$, and let $G(g,q) := \{g_{\gamma_i}\mid_{T_{\gamma_i}(q)}\}_{i\in I}$.
\end{itemize}
\end{definition}

\noindent Making a promise in a game $g$ by defining a quasi-profile $q$ splits the game into "smaller" games, formally \textit{via} $G(g,q)$. If the promise is rational, these "smaller" games can be processed independently since gluing any of their respective SPE will yield an SPE for $g$. Towards this, Observation~\ref{obs:tree-partition} below suggests that the recursive refinement will lead to a fully defined strategy profile of $g$, if performed a sufficiently great (ordinal) number of times.

\begin{observation}\label{obs:tree-partition}
Let $g$ be a game on a tree $T$, let $q$ be a quasi profile for $g$, and let $G(g,q) = \{g_{\gamma_i}\mid_{T_{\gamma_i}(q)}\}_{i\in I}$. Then $\{\gamma_iT_{\gamma_i}(q)\}_{i\in I} $ is a partition of $T$.
\end{observation}

\section{Many players with linearly ordered preferences}\label{sec:mp-lop}

\noindent This section characterizes the linear preferences that always yield SPE in games with $\mathcal{D}_{\omega_1}$-measurable outcome functions: the families of preferences without infinite ascending chains and without the SPE killer, \textit{i.e.} the pattern $z <_a y <_a x\,\wedge\, x <_b z <_b y$ for some players $a$ and $b$ and outcomes $x,y,z$. The main difficulty is tackled by Lemma~\ref{lem:quasi-antagonist-closed-SPE} and corollary~\ref{cor:quasi-antagonist-open-union-closed-SPE} below. It consists in slightly generalizing an existing result~\cite{LP14} stating that two-player Borel games with antagonist preferences have SPE by considering preferences that are almost antagonist, but in addition there is an outcome $y$ that is the worst one for both players, and the set of plays with outcome $y$ is a closed set (union an open set). This is then generalized for a set in $\mathcal{D}_{\omega_1}$ by induction, and eventually to multi-player games without the SPE killer thanks to a combinatorial result.  

\begin{lemma}\label{lem:quasi-antagonist-closed-SPE}
Let a game involve two players $a$ and $b$, preferences $y <_a x_n <_a \dots <_a x_1$ and $y <_b x_1 <_b \dots <_b x_n$ for some $n$, such that all plays without pseudo-leaves have outcome $y$. The game has a global-Pareto subgame perfect equilibrium.
\end{lemma}

\begin{proof}
Let us consider only infinite games involving two players $a$ and $b$, preferences $y <_a x_n <_a \dots <_a x_1$ and $y <_b x_1 <_b \dots <_b x_n$ for some $n$. Let us call a game \emph{weak-stop} if every play without pseudo-leaves has outcome $y$, and \emph{strong-stop} if in addition every node that does not lie on a play with outcome $y$ has a prefix that is a pseudo-leaf. Note the following: in every weak-stop game the plays with outcome $y$ form a closed set; for every quasi profile $q$ for a weak-stop game $g$, the set $G(g,q)$ contains only weak-stop games; and modifying the outcome function of a weak-stop game such that it is constant on given subgames yields a weak-stop game. (But the same does not hold for strong-stop games.) Let us call a node of a strong-stop game an $a_k$ ($b_k$) node if it is owned by player $a$ ($b$), and if it is the parent of a $x_k$-pseudo-leaf. Let us call every $a_k$ ($b_k$) node a $a$-stop ($b$-stop) node, and furthermore let us call every $a$-stop or $b$-stop node a stop node.

Let us prove the claim by induction on $n$, which holds for $n = 0$ and $n = 1$, so let us assume that $1 < n$. Five transformations on games are defined below, and they are meant to be applied recursively to a weak-stop game.

\begin{enumerate} 

\item\label{trans:1} "Weak-stop towards strong-stop": Let $g$ be a weak-stop game on tree $T$. Let $\gamma$ be a node such that $g_\gamma$ involves more than one outcome but not $y$. By construction $g_\gamma$ is an antagonist game, and it amounts, when seeing a pseudo-leaf as a leaf, to a game without infinite plays, but possibly without uniform bound on the length of the plays. By~\cite{EO10} it has a GP-SPE $s_\gamma$ nonetheless, which induces some $x_k$. Let us derive a weak-stop game $g'$ from $g$ by modification of the outcome function: for all $p\in [T]$ let $v'(p) := x_k$ if $\gamma \sqsubseteq p$ and $v'(p) := v(p)$ otherwise. Pasting $s_\gamma$ at node $\gamma$ on a GP-SPE $s'$ for $g'$ yields a GP-SPE $s$ for $g$. Formally $s(\gamma \delta) := s_\gamma(\delta)$ for all $\delta\in T_\gamma$ and $s(\delta) := s'(\delta)$ for all $\delta\in T$ such that $\gamma \not\sqsubseteq \delta$.

\item\label{trans:2} "Emptying the interior of $y$": Let $g$ be a strong-stop game on tree $T$ and let $\gamma$ be a $y$-pseudo-leaf. If $\gamma = \epsilon$ all profiles for $g$ are GP-SPE. If $\gamma$ is not the root of $g$ let us define a quasi profile $q$ for $g$ by letting the owner of the parent of $\gamma$ ignore $\gamma$. Formally, $q(\delta) := \mathrm{succ}(T,\delta)\backslash \{\gamma\}$ for all $\delta \in T$. Since $y$ is the worst outcome for both players, none will have an incentive to deviate from this promise, regardless of the future choices at the other nodes. $G(g,q)$ contains $g_\gamma$ and a  weak-stop game $g'$. Combining any profile $s_\gamma$ for $g_\gamma$ and a GP-SPE for $g'$ yields a GP-SPE for $g$.

\item\label{trans:3} "$b$ chooses $x_n$": Let $g$ be a strong-stop game on tree $T$, let $\gamma$ be a $b_n$ node and let $\delta\in \mathrm{succ}(T,\gamma)$ be an $x_n$-pseudo-leaf. Let us define a quasi profile $q$ for $g$ by letting $b$ choose $\delta$ at $\gamma$. Formally, $q(\gamma) := \{\delta\}$ and $q(\alpha) := \mathrm{succ}(T,\alpha)$ for all $\alpha\in T\backslash \{\gamma\}$. Since $x_n$ is $b$'s preferred outcome, she will have no incentive to deviate from this choice, regardless of the choices at the other nodes. So, finding a GP-SPE for every weak-stop game in $G(g,q)$ will complete the definition of a GP-SPE for $g$. 

\item\label{trans:4} "$a$ ignores $x_n$": Let $g$ be a strong-stop game. Let $\gamma$ be a node in $g$ such that $g_\gamma$ involves outcome $y$ but no $b$-stop nodes, and such that every subgame of $g_\gamma$ involving outcome $y$ has an $a_k$ node for some $k < n$. Let us define a quasi profile $q$ for $g$ by letting $a$ ignore all $x_n$-pseudo-leaves at all $a_n$ nodes below $\gamma$. The set $G(g,q)$ is made of games involving only outcome $x_n$ and of one $g'$ such that $g'_\gamma$ does not involve $x_n$. By induction hypothesis $g'_\gamma$ has an GP-SPE $s'_\gamma$, which induces some $x_k$. Let us define $g''$ by modification of the outcome function of $g$: for all $p\in \gamma[T_{\gamma}]$ let $v''(p) := x_k$ and for all $p\in [T]\backslash \gamma [T_{\gamma}]$ let $v''(p) := v(p)$. Pasting $s'_\gamma$ on a GP-SPE $s''$ for $g''$ yields a GP-SPE for $g$.

\item\label{trans:5} "$a$ chooses $x_n$": Let $g$ be a strong-stop game. Let $\gamma$ be such that $g_\gamma$ involves outcome $y$ and every subgame of $g_\gamma$ involving $y$ has some $a_n$ nodes but no $a_k$ node for all $k < n$. To build for $g_\gamma$ a GP-SPE $s_\gamma$ inducing $x_n$ on all of its subprofiles, it suffices, first, to choose arbitrary profiles for the subgames rooted at the $x_n$-pseudo-leaves of $g_\gamma$, and second, to fix consistently paths from each node to an $x_n$-pseudo-leaf. This second step can be done by letting player $a$ choose an $x_n$-pseudo-leaf at some node, which defines a quasi profile $q$, and by repeating it recursively for the games in $G(g_\gamma,q)$. This $s_\gamma$ is a GP-SPE because every subprofile induces outcome $x_n$, which is $b$'s preferred outcome, and the only alternative for $a$ in every subprofile is outcome $y$ since all $a$-stop nodes are $a_n$ nodes. Let us define $g'$ by modification of the outcome function of $g$: for all $p\in \gamma [T_{\gamma}]$ let $v'(p) := x_n$ and for all $p\in [T]\backslash \gamma [T_{\gamma}]$ let $v'(p) := v(p)$. Pasting $s_\gamma$ on a GP-SPE $s'$ for $g'$ yields a GP-SPE for $g$.
\end{enumerate}

Given a weak-stop game $g$, let us apply to it the five transformations above, sequentially, non-deterministically whenever they are applicable, and until none of them is applicable, \textit{i.e.} possibly an ordinal number of times. This yields a set $G$ of strong-stop games (otherwise Transformation~\ref{trans:1} could be applied) whose subgames that involve outcome $y$ all have stop nodes (otherwise Transformation~\ref{trans:2} could be applied), without $b_n$ nodes (otherwise Transformation~\ref{trans:3} could be applied), such that every subgame that involves $y$ but no $b$-stop nodes has a subgame without $a_k$ nodes for all $k < n$ (otherwise Transformation~\ref{trans:4} could be applied), and such that every subgame $h'$ of every game in $G$ has the following property (otherwise Transformation~\ref{trans:5} could be applied): if every subgame of $h'$ has a $a$-stop node, $h'$ has an $a_k$ node for some $k < n$.

Let $h'$ be a subgame of $h\in G$, and that involves $y$. If $h'$ has no $b$-stop nodes, combining the properties above shows that all of its subgames have $a$-stop nodes, so one of them has only $a_n$ nodes, contradiction, so every subgame of $h$ that involves $y$ has a $b$-stop node.

For every $h\in G$ let us define the quasi profile $q$ by letting $a$ ignore all the $x_n$-pseudo-leaves. $G(h,q)$ is made of games involving the outcome $x_n$ only and of one $h'$ void of $x_n$. Since every subgame of $h$ that involves $y$ has a $b$-stop node, it also holds for $h'$. By induction hypothesis, $h'$ has a GP-SPE, which is easily extended to a GP-SPE for $h$, thus completing the definition of a GP-SPE for the original $g$.
 \end{proof}

\begin{corollary}\label{cor:quasi-antagonist-open-union-closed-SPE}
Given a game $g$ with two players $a$ and $b$, a quasi-Borel measurable outcome function, and preferences $y <_a x_n <_a \dots <_a x_1$ and $y <_b x_1 <_b \dots <_b x_n$ for some $n$. If the plays with outcome $y$ form the union of an open set and a closed set, the game has a global-Pareto subgame perfect equilibrium.
\end{corollary}

\begin{proof}
Let the plays with outcome $y$ be the union $Y = Y_o \cup Y_c$ of  an open set and a closed set. Wlog $Y_o$ and $Y_c$ are disjoint. Let us derive $g'$ from $g$ by removing the plays in $Y_o$. So the plays of $g'$ that do not yield outcome $y$ form an open set, \textit{i.e.} a disjoint union of clopen balls with defined by the prefixes $(\gamma_i)_{i\in I}$. Every game $g_{\gamma_i}$ is antagonist and quasi-Borel, so it has an SPE $s_i$ by ~\cite{LP14}. Let us define $g''$ by modification of the outcome function of $g'$: every play going through $\gamma_i$ yields the outcome induced by $s_i$. This $g''$ has a GP-SPE $s''$ by Lemma~\ref{lem:quasi-antagonist-closed-SPE}, and together with the $s_i$ it can be used to build a GP-SPE for $g$.
\end{proof}

\noindent Let us extend Corollary~\ref{cor:quasi-antagonist-open-union-closed-SPE} from open union closed to the difference hierarchy.

\begin{lemma}\label{lem:quasi-antagonist-Delta02-SPE}
Let be a game on a tree $T$, with two players $a$ and $b$ and preferences $y <_a x_n <_a \dots <_a x_1$ and   $y <_b x_1 <_b \dots <_b x_n$, and let us assume that each set of plays with outcome $x_i$ is quasi-Borel and that $Y$ the set of plays with outcome $y$ is in $\mathcal{D}_{\omega_1}([T])$. Then the game has a global-Pareto subgame perfect equilibrium.
\end{lemma}

\begin{proof}
By transfinite induction on the level of $Y$ in the difference hierarchy. The base case where $Y$ is open or closed is solved by Corollary~\ref{cor:quasi-antagonist-open-union-closed-SPE}.  

For the inductive case, let us make a case disjunction depending on the last step of the construction of $Y$. For the union case, $Y = \cup_{i\in \mathbb{N}} \gamma_iY_i$ for some $Y_i$ that have lower levels than $Y$ in the difference hierarchy, and where the $\gamma_i$ are not prefixes of one another. By induction hypothesis each $g_{\gamma_i}$ has a GP-SPE $s_i$ inducing either $y$ or an outcome $x_{k(i)}$. Let us start the construction of a profile $s$ for $g$ by fixing the $s_i$ as the respective subprofiles for the $g\mid_{\gamma_iT_{\gamma_i}}$. Let us define $g'$ by modification of the outcome function of $g$: let each play going through $\gamma_i$ yield the outcome induced by $s_i$. This is a quasi-Borel game and the plays with outcome $y$ form an open set, so it has a GP-SPE $s'$ by Corollary~\ref{cor:quasi-antagonist-open-union-closed-SPE}, which we use to complete the definition of $s$. It is easy to check that $s$ is a GP-SPE for $g$.

For the complementation case, $Y = [T]\backslash ([T]\backslash Y)$, where $[T]\backslash Y$ is equal to $\cup_{i\in \mathbb{N}} \gamma_iX_i$ for some $X_i$ that have lower levels than $[T]\backslash Y$ (and $Y$) in the difference hierarchy, and where the $\gamma_i$ are not prefixes of one another. Since all $Y \cap [\gamma_iT_{\gamma_i}] = [\gamma_iT_{\gamma_i}]\backslash \gamma_iX_i$ have lower levels than $Y$ in the difference hierarchy, by induction hypothesis each $g_{\gamma_i}$ has a GP-SPE $s_i$ inducing either $y$ or an outcome $x_{k(i)}$. Let us start the construction of a profile $s$ for $g$ by fixing the $s_i$ as the respective subprofiles for the $g\mid_{\gamma_iT_{\gamma_i}}$. Let us define $g'$ by modification of the outcome function of $g$: let each play going through $\gamma_i$ yield the outcome induced by $s_i$. This is a quasi-Borel game and the plays with outcome $y$ form the union of an open set and the closed set $[T]\backslash \cup_{i\in I} [\gamma_iT_{\gamma_i}]$, so it has a GP-SPE $s'$ by Corollary~\ref{cor:quasi-antagonist-open-union-closed-SPE}, which we use to complete the definition of $s$. It is easy to check that $s$ is a GP-SPE for $g$.
\end{proof}

\noindent The combinatorial Lemma~\ref{lem:forbid-pattern} below shows that the "local" absence of the SPE-killer amounts to a very simple "global" structure.

\begin{lemma}\label{lem:forbid-pattern}
Let $A$ be a non-empty set and for all $a\in A$ let $<_a$ be a strict linear order over some non-empty set $O$. The following assertions are equivalent.
\begin{enumerate}
\item\label{eq:pattern} $\forall a,b\in A,\forall x,y,z\in O,\,\neg(z<_ay<_ax\,\wedge\,x<_bz<_by)$.
\item\label{eq:partition} There exists a partition $\{O_i\}_{i\in I}$ of $O$ and a linear order $<$ over $I$ such that:
\begin{enumerate}
\item\label{cond:interval} $i<j$ implies $x<_ay$ for all $a\in A$ and $x\in O_i$ and $y\in O_j$.
\item\label{cond:inverse} $<_b\mid_{O_i}=<_a\mid_{O_i}$ or $<_b\mid_{O_i}=<_a\mid_{O_i}^{-1}$ for all $a,b\in A$.
\end{enumerate}
\end{enumerate}
If $1.$ and $2.$ hold, we may also assume that $<_b\mid_{O_i}=<_a\mid_{O_i}^{-1}$ is witnessed for all $i\in I$. Also, $O_i$ is always a $<_a$-interval for all $(i,a)\in I \times A$, as implied by \ref{cond:interval}.
\end{lemma}

\begin{proof}
\ref{eq:partition}$\,\Rightarrow\,$\ref{eq:pattern} is straightforward, so let us assume \ref{eq:pattern}. Let $x\sim y$ stand for $\exists a,b\in A,\,x\leq_a y\leq_b x$, which defines a reflexive and symmetric relation, and note that due to the SPE killer $z<_ay<_ax\,\wedge\,x<_bz<_by$ the following holds: if $x<_ay$ and $y<_bx$, then $x<_a z<_a y$ iff $y<_b z<_b x$. To show that $\sim$ is transitive too, let us assume that $x\sim y\sim z$. If $x,y,z$ are not pairwise distinct, $x\sim z$ follows directly, so let us assume that they are pairwise distinct, so by assumption there exist $a,b,c,d\in A$ such that $y<_a x<_b y<_c z<_d y$. To show that $x\sim z$ there are three cases depending on where $z$ lies with respect to $y<_ax$, all cases invoking the (above-mentioned) forbidden-pattern argument: if $y<_a z<_a x$ then $x<_b z<_b y$, and $x\sim z$ follows; if $y<_a x<_a z$ then $z<_d x<_d y$ and subsequently $y<_c x<_c z$, by invoking twice the forbidden-pattern argument, and $x\sim z$ follows; third case, let us assume that $z<_a y<_a x$. If $x<_b z$ then $x\sim z$ follows, and if $z<_b x$ then $z<_b x<_by$, so $y<_c x<_c z$, and $x\sim z$ follows. Therefore $\sim$ is an equivalence relation; let $\{O_i\}_{i\in I}$ be the corresponding partition of $O$. 

Now let us show that the  $\sim$-classes are $<_a$-intervals for all $a$, so let $x\sim y$ and $x<_a z<_a y$. By definition of $\sim$, there exists $b$ such that $y<_b x$, in which case $y<_b z<_bx$ by the forbidden-pattern argument, so $x\sim z$ by definition. 

Let $x\in O_i$ and $y\in O_j$ be such that $x<_ay$. Since $O_i$ and $O_j$ are intervals, $x'<_ay'$ for all $x'\in O_i$ and $y'\in O_j$. Since $\neg(x'\sim y')$ by assumption, $x'<_by'$ for all $b\in A$, by definition of $\sim$. In this case defining $i<j$ meets the requirements. 

Before proving \ref{cond:inverse} let us prove that if $x<_ay$ and $y<_bx$ and $z\sim y$, then $z<_ay$ iff $y<_bz$: this is trivial if $z$ equals $x$ or $y$, so let us assume that $x\neq z\neq y$, and also that $z<_ay$. If $z<_by$, then $y<_cz$ for some $c$ since $z\sim y$, and wherever $x$ may lie with respect to $y<_cz$, it always yields a SPE killer using $<_a$ or $<_b$, so $y<_bz$. The converse is similar, it follows actually from the application of this partial result using $<_b^{-1}$ and $<_a^{-1}$ instead of $<_a$ and $<_b$.

Now assume that $<_b\mid_{O_i}\neq<_a\mid_{O_i}$ for some $O_i$, so $x<_ay$ and $y<_bx$ for some $x,y\in O_i$. Let $z,t\in O_i$. By the claim just above $z<_ay$ iff $y<_bz$, so by the same claim again $z<_at$ iff $t<_bz$, which shows that $<_b\mid_{O_i}=<_a\mid_{O_i}^{-1}$. This proves the equivalence.

Finally, let us assume that the assertions hold. By definition of $\sim$, if $O_i$ is not a singleton, $x<_ay$ and $y<_bx$ for some $a,b\in A$ and $x,y\in O_i$, so $<_b\mid_{O_i}=<_a\mid_{O_i}^{-1}$ is witnessed.
\end{proof}

\noindent Theorem~\ref{thm:diff-hier-GP-SPE} extends Lemma~\ref{lem:quasi-antagonist-Delta02-SPE} to many players and more complex preferences.

\begin{theorem}\label{thm:diff-hier-GP-SPE}
Let $g$ be a quasi-Borel game with players in $A$, outcomes in $O$, and linear preferences $<_a$ for all $a\in A$. Let us assume that the inverses of the $<_a$ are well-ordered, and that there exists a partition $\{O_i\}_{i\in I}$ of $O$ and a linear order $<$ of $I$ such that:
\begin{itemize}
\item $i<j$ implies $x<_ay$ for all $a\in A$ and $x\in O_i$ and $y\in O_j$.
\item $<_b\mid_{O_i}=<_a\mid_{O_i}$ or $<_b\mid_{O_i}=<_a\mid_{O_i}^{-1}$ for all $a,b\in A$ and $i\in I$.
\end{itemize}
Let us further assume that for all $i\in I$ the plays with outcome in $\cup_{j < i} O_j$ form a $\mathcal{D}_{\omega_1}$ set. Then $g$ has a global-Pareto subgame perfect equilibrium.
\end{theorem}

\begin{proof}
By Lemma~\ref{lem:forbid-pattern} let us further assume wlog that $<_b\mid_{O_j}=<_a\mid_{O_j}^{-1}$ is witnessed for all $j\in I$, so the $O_j$ are finite by well-ordering. Let us build a GP-SPE for $g$ as the limit of a recursive procedure: Let $O_i$ be such that some outcome of $O_i$ occurs in $g$ and such that for all $j > i$ no outcome from $O_j$ occurs in $g$. Let $(\gamma_k)_{k\in K}$ be the shortest nodes of $g$ such that the outcomes occuring in $g_{\gamma_k}$ are in $\cup_{j < i}O_j$ only. Let us define a quasi profile $q$ for $g$ by having the $\gamma_k$ ignored by their parents. $G(g,q)$ consists of the $g_{\gamma_k}$ and of a game $g'$. By Lemma~\ref{lem:quasi-antagonist-Delta02-SPE} there is a GP-SPE $s'$ for $g'$. (To see this, replace $\cup_{j < i} O_j$ with one single fresh outcome $y$, \textit{i.e.} $y\notin O$ and set $y <_a x$ for all $a\in A$ and $x\in O_i$.)  Combining $s'$ with GP-SPE for the $g_{\gamma_k}$ (obtained recursively) yields a GP-SPE for $g$, since the choices made in $g'$ hold regardless of the choices made in the $g_{\gamma_k}$.
\end{proof}

\begin{corollary}\label{cor:SPE-pattern-lin-eq}
Let $A$ and $O$ be non-empty finite sets (of players and of outcomes) and for all $a\in A$ let $<_a$ be a linear preference. The following are equivalent.
\begin{enumerate}
\item $\forall a,b\in A,\forall x,y,z\in O,\,\neg(z<_ay<_ax\,\wedge\,x<_bz<_by)$.
\item Every $\mathcal{D}_{\omega_1}$-Gale-Stewart game using $A$, $O$ and the $<_a$ has a GP-SPE.
\end{enumerate}

\begin{proof}
For $1.\Rightarrow 2.$ invoke Lemma~\ref{lem:forbid-pattern} and Theorem~\ref{thm:diff-hier-GP-SPE}, and prove $2.\Rightarrow 1.$ by contraposition with the following folklore example which is detailed, \textit{e.g.}, in~\cite{LP14}.

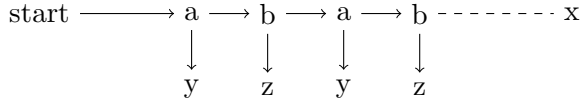
\begin{tikzpicture}[node distance=1cm]
  \node(s){start};
  \node(s1)[right of = s]{};
  \node(a1)[right of = s1]{a};
  \node(b1)[right of = a1]{b};
  \node(a2)[right of = b1]{a};
  \node(b2)[right of = a2]{b};
  \node(inf1)[right of = b2]{};
  \node(inf)[right of = inf1]{x};

   \node(y1)[below of = a1]{y};
   \node(z1)[below of = b1]{z};
   \node(y2)[below of = a2]{y};
   \node(z2)[below of = b2]{z};
    
  \draw [->] (s) -- (a1);
  \draw [->] (a1) -- (b1);
  \draw [->] (b1) -- (a2);
  \draw [->] (a2) -- (b2); 
  \draw [->] (a1) -- (y1);
  \draw [->] (b1) -- (z1);
  \draw [->] (a2) -- (y2); 
  \draw [->] (b2) -- (z2); 
  \draw [dashed] (b2) -- (inf);
  \end{tikzpicture}
\end{proof}
\end{corollary}

\noindent Corollary~\ref{cor:SPE-pattern-lin-eq} and the results that lead to it considers linear preference only. Proposition~\ref{prop:swo-lin-ext-SPE-killer} below show that this restriction incurs a loss of generality, which is partly solved in Section~\ref{sec:tp-swop}.

\begin{proposition}\label{prop:swo-lin-ext-SPE-killer}
Let us define two binary relations by $z,t \prec_a x,y$ and $y \prec_b z \prec_b x \prec_b t$.
\begin{enumerate}
\item $\mathcal{D}_{\omega_1}$ infinite games with players $a$ and $b$ and preferences $\prec_a$ and $\prec_b$ have SPE. 
\item The SPE killer occurs in any strict linear extensions of $\prec_a$ and $\prec_b$.
\end{enumerate}
\begin{proof}
$1.$ follows Theorem~\ref{thm:2player-swo-pref-SPE}. For $2.$ let $\prec'_a$ be a linear extension of $\prec_a$. If $x \prec'_a y$ then $z \prec'_a x \prec'_a y$ and $y \prec_b z \prec_b x$. If $y \prec'_a x$ then $t \prec'_a y \prec'_a x$ and $y \prec_b x \prec_b t$.
\end{proof}
\end{proposition}

\section{Two players with strict weak order preferences}\label{sec:tp-swop}

\noindent The preferences considered in Proposition~\ref{prop:swo-lin-ext-SPE-killer} are strict weak orders. Informally, strict weak orders are strict partial orders that can be seen as strict linear orders from afar, \textit{i.e.} up to an equivalence relation. Traditionally in game theory the outcomes are real-valued payoff functions $f,g: A\to\mathbb{R}$ and the preferences are defined by $f \prec_a g$ iff $f(a) < g(a)$. These preferences are not strict linear orders but they are strict weak orders, so the results from Section~\ref{sec:mp-lop} are worth generalizing. Strict weak orders are defined below.

\begin{definition}[Strict weak order]
A strict weak order is a strict partial order whose complement is transitive, \textit{i.e.} is satisfies $\neg(x \prec x)$ and $x \prec y \wedge y \prec z \Rightarrow x \prec z$ and $\neg(x \prec y) \wedge \neg(y \prec z) \Rightarrow \neg(x \prec z)$.
\end{definition}

\noindent Lemma~\ref{lem:forbid-pattern} above describes the structure of strict linear orders void of the SPE killer. A similar result for strict weak orders will be useful. Lemma~\ref{lem:swo-antagonistic} below is part to it, and the other part appears directly in the proof of Lemma~\ref{lem:2p-swo-spe}.

\begin{lemma}\label{lem:swo-antagonistic}
Let $\prec_a$ and $\prec_b$ be two strict weak orders over some finite $O$.
\begin{enumerate}
\item If $\prec_a$ and $\prec_b$ are void of the SPE killer, if there exists a $\prec_a$-non-extremal element and a $\prec_b$-non-extremal element, and if there is no partition $\{O_u,O_l\}$ of $O$ such that $\neg(x \prec_a y)$ and $\neg(x \prec_b y)$ for all $(x,y)\in O_u \times O_l$, then $ \prec_a \cap \prec_b = \emptyset$.

\item If $\prec_a \cap \prec_b = \emptyset$, there exists a linear extension $<$ of $\prec_a$ such that $\prec_b \subseteq <^{-1}$.
\end{enumerate}
\begin{proof}
\begin{enumerate}
\item Let $x$ be $\prec_b$-minimal among the $\prec_a$-maximal elements, and let $y$ be $\prec_a$-minimal among the $\prec_b$-maximal elements, so $x \neq y$ by the partition assumption. Towards a contradiction let us assume that, \textit{e.g.}, $x$ is not $\prec_b$-minimal, and let $z$ be  $\prec_b$-minimal. For all $t \prec_b x$, it follows that $t$ is not $\prec_a$-maximal by definition of $x$, and $\neg(y \prec_a t)$ by absence of the SPE killer. So $y$ is not $\prec_a$-minimal, otherwise $z$ is also $\prec_a$-minimal, thus contradicting the partition assumption. So likewise, for all $t \prec_a y$, it follows that $t$ is not $\prec_b$-maximal  by definition of $y$, and $\neg(x \prec_b t)$ by absence of the SPE killer. So the two-element partition induced by $\{t\in O\,\mid\, t\prec_a y\,\vee\, t\prec_b x\}$ contradicts the partition assumption. This shows that $x$ is $\prec_b$-minimal and $y$ is $\prec_a$-minimal. Towards a contradiction let us assume that $t (\prec_a\cap \prec_b) z$ for some $t,z\in O$. So $\{x,y\}\cap\{z,t\} = \emptyset$. By the partition assumption $z$ is not both $\prec_a$ and $\prec_b$-maximal, so, \textit{e.g.}, $z \prec_a x$, and $t$ is not both $\prec_a$ and $\prec_b$-minimal. By absence of the SPE killer $t$ is $\prec_b$-minimal, so $y \prec_a t$ by the partition assumption., and subsequently $z$ is $\prec_b$-maximal. By assumption there exists $\gamma$ that is neither $\prec_b$-maximal nor $\prec_b$-minimal. Wherever $\gamma$ lies wrt $\prec_a$, the SPE killer occurs.

\item By induction on the cardinality of $O$, which holds for $|O| = 0$. Let $x$ be $\prec_b$-minimal among the $\prec_a$-maximal elements, so $x$ is also $\prec_b$-minimal since $\prec_a \cap \prec_b = \emptyset$. By induction hypothesis let $<_x$ witness the claim for $\prec_a\mid_{O\backslash\{x\}}$ and $\prec_b\mid_{O\backslash\{x\}}$. The linear order $< \,:=\, <_x \cup \{(y,x)\,\mid\,y\in O\backslash\{x\}\}$ witnesses the claim.
\end{enumerate}
\end{proof}
\end{lemma}

\noindent Lemma~\ref{lem:2p-swo-spe} below is a generalization of Corollary~\ref{cor:quasi-antagonist-open-union-closed-SPE} from an order-theoretic point of view and a special case thereof from a  topological point of view. Due lack of space and strong similarities with the proof of Lemma~\ref{lem:quasi-antagonist-closed-SPE}, the proof of Lemma~\ref{lem:2p-swo-spe} is in appendix.

\begin{lemma}\label{lem:2p-swo-spe}
Let $g$ be a game with two players $a$ and $b$, finitely many outcomes $O$, and strict weak order preferences void of the SPE killer. If each outcome corresponds to the union of an open set and a closed set, the game has an SPE.

Furthermore, for every node $\gamma$ of $g$ let $\{O^\gamma_1,\dots,O^\gamma_{n_{\gamma}}\}$ be a partition of the outcomes of $g_\gamma$ such that $\neg(x \prec_a y)$ and $\neg(x \prec_b y)$ for all $1\leq k < n_\gamma$ and $(x,y)\in O^\gamma_{k+1}\times O^\gamma_k$. There exists an SPE for $g$ such that the outcome induced at every node $\gamma$ belongs to $O^{\gamma}_{n_{\gamma}}$.

\begin{proof}
By induction on the number of outcomes, which holds up to two outcomes. Let us make a case disjunction for the inductive case. First main case, there is no partition $\{O_u,O_l\}$ of $O$ such that $\forall (x,y)\in O_u \times O_l, \neg(x \prec_a y)\,\wedge\, \neg(x \prec_b y)$. Let us make a nested case distinction. First nested case, there exists a $\prec_a$-non-extremal element and a $\prec_b$-non-extremal element, so by Lemma~\ref{lem:swo-antagonistic} let $<$ be a linear extension of $\prec_a$ such that $\prec_b \subseteq <^{-1}$. By~\cite{LP14} the antagonist game with preference $<$ has an SPE, which is also an SPE for $\prec_a$ and $\prec_b$.

Second nested case, one preference, \textit{e.g.}, $\prec_a$ has only extremal elements. By the partition assumption let $y$ be $\prec_b$-maximal and $\prec_a$-minimal. Let $Y$ be the set of plays with outcome $y$, and let us define a quasi-profile $q$ as follows. Let $\gamma$ be the parent of a $y$-pseudo-leaf. If $\gamma$ is owned by $a$, let $a$ ignore the $y$-pseudo-leaves at $\gamma$; otherwise let $b$ choose a $y$-pseudo-leaf at $\gamma$. Let us apply this construction recursively (an ordinal number of times) to the games in $G(g,q)$ that do not involve $y$, until the $Y'$ of each remaining game has empty interior. It is easy to check that $Y'$ is also closed since $Y$ is the union of an open set and a closed set, by assumption. So let $(\gamma_i)_{i\in I}$ be the shortest nodes that are not on any play with outcome $y$. The $g_{\gamma_i}$ do not involve $y$, so by induction hypothesis they have suitable SPE $s_i$ inducing some $x_i\in O$, which allow us to start the definition of a suitable SPE for $g$. Let us define $g'$ by modification of the outcome function of $g$: let $v'(\gamma_i p) := x_i$ for all $i\in I$ and let $v'(p) := v(p)$ when $\gamma_i\not\sqsubseteq p$ for all $i\in I$. Let $M_a$ be the $\prec_a$-maximal outcomes, and let us define a quasi-profile $q'$ for $g'$ as follows. Let $\gamma$ be the parent of a $M_a$-pseudo-leaf, \textit{i.e.} a shortest node involving outcome in $M_a$ only. If $\gamma$ is owned by $a$, let $a$ choose a $M_a$-pseudo-leaf at $\gamma$; otherwise let $b$ ignore the $M_a$-pseudo-leave at $\gamma$. Let us apply this recursively to the games in $G(g',q)$ that involve $y$. In the remaining games $a$ can deviate from a play with outcome $y$ only to reach a $(O\backslash)M_a$-pseudo-leaf, so every profile that  follows plays with outcome $y$ whenever possible is a suitable SPE.



Second main case, there exist partitions $\{O_u,O_l\}$ of $O$ such that $\forall (x,y)\in O_u \times O_l, \neg(x \prec_a y)\,\wedge\, \neg(x \prec_b y)$. Among these partitions let us consider the one with the smallest possible $O_u$. Let us make a further case disjunction. First case, $|O_l| > 1$. As is now customary, let us start defining a suitable SPE for $g$ by using the induction hypothesis on the maximal subgames involving only outcomes in $O_l$, and on the game derived from $g$ by replacing outcomes in $O_l$ with a fresh outcome $y$ that is the new $\prec_a$ and $\prec_b$-minimum.  

Second case, $|O_l| = \{y\}$. By minimality of $|O_u|$, there is no partition $\{O_{uu},O_{ul}\}$ of $O_u$ such that $\forall (x,y)\in O_{uu} \times O_{ul}, \neg(x \prec_a y)\,\wedge\, \neg(x \prec_b y)$. Therefore the situation is reminiscent of the first main case above, but for $\prec_a\mid_{O_u}$ and $\prec_b\mid_{O_u}$ instead of $\prec_a$ and $\prec_b$. In both nested cases from the first main case, there exists some $x_n$ that is, \textit{e.g.}, $\prec_a\mid_{O_u}$-minimal and $\prec_b\mid_{O_u}$-maximal. Applying the proof of Lemma~\ref{lem:quasi-antagonist-closed-SPE} almost verbatim yields a suitable SPE for $g$

\end{proof}
\end{lemma}

\begin{theorem}\label{thm:2player-swo-pref-SPE}
Let $g$ be a game with two players $a$ and $b$, finitely many outcomes, a $\mathcal{D}_{\omega_1}$-measurable outcome function, strict weak order preferences such that $\neg(z \prec_a y \prec_a x \,\wedge\, x \prec_b z \prec_b y)$ for all outcomes $x$, $y$ and $z$. Then the game has an SPE. 

Furthermore, for every node $\gamma$ of $g$ let $\{O^\gamma_1,\dots,O^\gamma_{n_{\gamma}}\}$ be a partition of the outcomes of $g_\gamma$ such that $\neg(x \prec_a y)$ and $\neg(x \prec_b y)$ for all $1\leq k < n_\gamma$ and $(x,y)\in O^\gamma_{k+1}\times O^\gamma_k$. There exists an SPE for $g$ such that the outcome induced at every node $\gamma$ belongs to $O^{\gamma}_{n_{\gamma}}$.

\begin{proof}
By induction on the levels in the difference hierarchy of the sets of plays corresponding to the outcomes. The base case holds by Lemma~\ref{lem:2p-swo-spe}.

For the inductive case, let $y$ be an outcome whose corresponding set $Y$ has level more than one in the difference hierarchy, and let us make a case disjunction depending on the last step of the construction of $Y$. The remainder of the proof can be taken almost verbatim from the proof of Lemma~\ref{lem:quasi-antagonist-Delta02-SPE}, but by replacing "GP-SPE" with "suitable SPE", and by invoking Lemma~\ref{lem:2p-swo-spe} or the induction hypothesis instead of Corollary~\ref{cor:quasi-antagonist-open-union-closed-SPE}.
\end{proof}
\end{theorem}

\noindent Theorem~\ref{thm:2player-swo-pref-SPE} considers two-player games only. Observation~\ref{obs:3player-swo} shows that absence of the SPE killer is no longer a sufficient condition for a three-player game with strict weak order preferences to have an SPE.

\begin{observation}\label{obs:3player-swo}
Let three players $a$,$b$ and $c$ have preferences $z \prec_a y \prec_a x$ and $t \prec_b z \prec_b y$ and $x \prec_c t \prec_c y$. (and, \textit{e.g.}, $y \sim_a t$, $z \sim_b x$, and $y \sim_c z$ or $x \sim_c z$)
\begin{enumerate}
\item The SPE killer does not occur in the strict weak orders $\prec_a$ and $\prec_b$ and $\prec_c$.
\item\label{obs:3player-swo2} The following game with $\prec_a$ and $\prec_b$ and $\prec_c$ has no SPE.

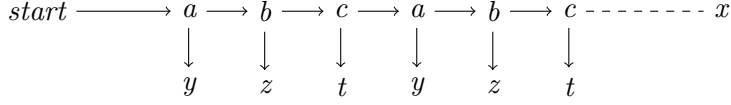
\begin{tikzpicture}[node distance=1cm]
  \node(s){start};
  \node(s1)[right of = s]{};
  \node(a1)[right of = s1]{a};
  \node(b1)[right of = a1]{b};
  \node(a2)[right of = b1]{c};
  \node(b2)[right of = a2]{a};
  \node(a3)[right of = b2]{b};
  \node(b3)[right of = a3]{c};
  
  \node(inf1)[right of = b3]{};
  \node(inf)[right of = inf1]{x};

   \node(y1)[below of = a1]{y};
   \node(z1)[below of = b1]{z};
   \node(y2)[below of = a2]{t};
   \node(z2)[below of = b2]{y};
   \node(y3)[below of = a3]{z};
   \node(z3)[below of = b3]{t};

  \draw [->] (s) -- (a1);
  \draw [->] (a1) -- (b1);
  \draw [->] (b1) -- (a2);
  \draw [->] (a2) -- (b2);
  \draw [->] (b2) -- (a3);
  \draw [->] (a3) -- (b3);  
  \draw [->] (a1) -- (y1);
  \draw [->] (b1) -- (z1);
  \draw [->] (a2) -- (y2); 
  \draw [->] (b2) -- (z2); 
  \draw [->] (a3) -- (y3); 
  \draw [->] (b3) -- (z3);   
   
  \draw [dashed] (b3) -- (inf);
  \end{tikzpicture}
\end{enumerate}

\begin{proof}
For~\ref{obs:3player-swo2}. Towards a contradiction let us assume that there exists an SPE for the game. Let us consider a node where player $a$ chooses $y$. Then at the node right above it $c$ chooses to continue to benefit from $y$, and at the node above $b$ chooses to continue, too. The induced outcome at the node further above is $y$ regardless of the choice of $a$, and so on up to the root.

Let us make a case disjunction: first case, there exists infinitely many nodes where $a$ chooses $y$, so $b$ and $c$ always continue by the remark above, so $a$ has an incentive to continue too, to induce outcome $x$, contradiction. Second case, there exists a node  below which $a$ always continues. From then on, one player must stop at some point, otherwise the outcome is $x$ and $c$ has an incentive to stop. The first player to stop cannot be $b$, otherwise $a$ would stop before $b$, and it cannot be $c$, otherwise $b$ would stop before $c$, contradiction. 
\end{proof}
\end{observation}

\noindent Proposition~\ref{prop:po-SPE-le-skp} below shows that considering only strict weak orders incurs a loss of generality. 

\begin{proposition}\label{prop:po-SPE-le-skp}
Let us define two binary relations by $\gamma \prec_a y \prec_a x$ and $z \prec_a \beta \prec_a \alpha$ and $x \prec_b z \prec y$ and $\alpha \prec_b \gamma \prec_b \beta$.
\begin{enumerate}
\item The SPE killer occurs in every strict weak order extensions of $\prec_a$ and $\prec_b$.
\item\label{prop:po-SPE-le-skp2} $\mathcal{D}_{\omega_1}$-games with players $a$ and $b$ and preferences $\prec_a$ and $\prec_b$ have SPE. 
\end{enumerate}

\begin{proof}
\begin{enumerate}
\item Let $\prec'_a$ be a strict weak order extension of $\prec_a$. If $z \prec'_a y$, the SPE killer occurs, with $x$. If $\neg(z \prec'_a y)$, then $\gamma \prec'_a \beta$ and the SPE killer occurs, with $\alpha$.
\item It suffices to prove the claim for games where each outcome set is the union of an open set and a closed set. (Then using a transfinite induction as in the proof of Lemma~\ref{lem:quasi-antagonist-Delta02-SPE} will do.) The techniques from Lemmas~\ref{lem:quasi-antagonist-closed-SPE} and \ref{lem:2p-swo-spe} are suitable here, and used without details. If the outcome $x$ does not occur in the game, note that $\prec_a\mid_{O\backslash \{x\}}$ and $\prec_b\mid_{O\backslash \{x\}}$ can be extended into the strict weak orders $z \prec'_a \gamma \sim'_a \beta \prec'_a y \sim'_a \alpha$ and $z \sim'_b \alpha \prec'_b \gamma \prec'_b y \sim'_b \beta$, respectively, and that the SPE killer is absent from these. So by Theorem~\ref{thm:2player-swo-pref-SPE} the game has an SPE wrt $\prec'_a$ and $\prec'_b$ and therefore also wrt $\prec_a$ and $\prec_b$. (And likewise if the outcome $\alpha$ does not occur in the game.)

Now one can reduce the set for $x$ to a closed set by letting $a$ choose the clopen balls with constant outcome $x$ and by letting $b$ ignore them. The plays with outcomes different from $x$ can be seen as belonging to a union of subgames without $x$, so by the remark above they have SPE. It allows us to replace these subgames with (pseudo)-leaves with outcomes the ones induced by the SPE. Now one can let $a$ choose the pseudo-leaves with outcome $\alpha$ and $b$ ignore them, which yields a game without outcome $\alpha$. So by the remark above there is an SPE for the game.
\end{enumerate}
\end{proof}
\end{proposition}

\section*{Acknowledgements}

I thank Vassilios Gregoriades and Arno Pauly for useful discussions. The author is supported by the ERC inVEST (279499) project.

\bibliographystyle{plain}
\bibliography{article}

\end{document}